\documentclass[aps,prl,twocolumn,superscriptaddress,showpacs,10pt]{revtex4}

\usepackage[T1]{fontenc}
\usepackage{times}
\usepackage{color,graphicx,xcolor}
\usepackage[latin1]{inputenc}
\usepackage[driverfallback=dvipdfm,colorlinks=true,linkcolor=blue,anchorcolor=blue,citecolor=red,urlcolor=magenta]{hyperref}
\usepackage{amsthm,amssymb,amsmath}

\newcommand\ket[1]{\ensuremath{|#1\rangle}}
\newcommand\bra[1]{\ensuremath{\langle#1|}}

\newcommand{\ketbra}[2]{| #1 \rangle\langle #2 |}

\newcommand\tr{\mathop{\rm tr}\nolimits}

\def\cS{{\cal S}}
\def\cT{{\cal T}}
\def\cW{{\cal W}}

\def\w{\mathbf{w}}

\newcommand{\bes}{\begin{eqnarray*}}
    \newcommand{\ees}{\end{eqnarray*}} 
\newcommand{\bpm}{\begin{pmatrix}}
    \newcommand{\epm}{\end{pmatrix}}

\def\IR{{\mathbb R}}

\def\cS{{\mathcal S}}
\def\cI{{\mathcal I}}

\def\cD{{\mathcal D}}

\def\cV{{\mathcal V}}

\def\diag{{\rm diag}\,}
\def\tr{{\rm tr}\,}

\newcommand{\defeq}{\stackrel{\smash{\textnormal{\tiny def}}}{=}}
 
\newcommand{\knorm}[2]{\lVert#1\rVert_{(#2)}}
\newcommand{\knormd}[2]{\lVert#1\rVert_{(#2)}^\circ}

\newtheorem{lemma}{Lemma}
\newtheorem{theorem}{Theorem}

\begin{document}

\title{Some notes on the robustness of $\mathbf{k}$-coherence and $\mathbf{k}$-entanglement}
\author{Nathaniel Johnston}
\affiliation{Department of Mathematics and Computer Science, Mount Allison University, Sackville, NB, Canada E4L 1E4}
\affiliation{Department of Mathematics and Statistics, University of Guelph, Guelph, ON, Canada N1G 2W1}
\author{Chi-Kwong Li}
\affiliation{Department of Mathematics, College of William and Mary,  Williamsburg, VA, USA  23187}
\author{Sarah Plosker}
\affiliation{Department of Mathematics \& Computer Science, Brandon University, Brandon,
    MB, Canada R7A 6A9}
    \affiliation{Department of Mathematics and Statistics, University of Guelph, Guelph, ON, Canada N1G 2W1}
\author{Yiu-Tung Poon}
\affiliation{Department of Mathematics, Iowa State University, Ames, IA, USA  50011}
\author{Bartosz Regula}
\affiliation{School of Mathematical Sciences, University of Nottingham, University Park, Nottingham NG7 2RD, United Kingdom}

\begin{abstract}
 We show that two related measures of $k$-coherence, called the standard and generalized robustness of $k$-coherence, are equal to each other when restricted to pure states. As a direct application of the result, we establish an equivalence between two analogous measures of Schmidt rank $k$-entanglement for all pure states. This answers conjectures raised in the literature regarding the evaluation of the quantifiers, and facilitates an efficient quantification of pure-state resources by introducing computable closed-form expressions for the two measures.
\end{abstract}

\date{\today}


\maketitle

\section{I. Introduction}
The degree to which a  quantum state (density matrix)
is in superposition with respect to a given set of mutually orthogonal states of the Hilbert space representing a  quantum system is called the \emph{coherence} of the state. Coherence has long since been established as a resource in quantum optics \cite{Glau63, Su63} that can be generated and detected. A general resource theory for coherence in quantum information theory has since been developed mirroring that of entanglement \cite{Abe06,BCP14,WY}, and for any measure of entanglement, one can analogously define a measure of coherence. Desirable properties have been identified for characterizing proper measures of coherence---they should equal zero precisely when the state is diagonal in the reference basis (such diagonal states are called \emph{incoherent}), they should be monotonic under incoherent quantum channels and under selective measurements on average, and they should be non-increasing under mixing of quantum states \cite{BCP14}.

A generalization of coherence, called \emph{$k$-coherence} (appearing in the literature under various guises: quantification strength of the quantumness of the state \cite{SV}, multilevel nonclassicality \cite{KSP, RPCBSA}, superposition rank \cite{TKEP}, and coherence number \cite{Chin1, Chin2}), provides a hierarchical structure for categorizing coherence. A state is $1$-incoherent if and only if it is incoherent (i.e., diagonal in the reference basis), all states are $n$-incoherent (where $n$ is the dimension of the Hilbert space the states act on), and for $1 < k < n$, a state being $k$-incoherent means that it is a convex combination of block diagonal matrices with blocks of size no larger than $k \times k$ (which intuitively corresponds to the state being ``more incoherent'' the smaller $k$ is).

Numerous proper measures of coherence have been identified and studied recently, such as the $\ell_1$-norm of coherence, the relative entropy of coherence \cite{BCP14}, and the robustness of coherence \cite{NBCPJA16}. These measures have also been generalized to $k$-coherence; see \cite{Reg17,RBCLAWFP17}. Here, we are interested in two separate generalizations of the robustness of coherence, which are called the ``standard'' and ``generalized'' robustnesses of $k$-coherence. Our main contribution is to show that the two measures agree with each other when restricted to pure-state inputs, which we do by deriving an explicit closed expression for the standard robustness that agrees with the formula for the generalized robustness obtained in~\cite{Reg17}. The equality between the robustnesses is not a priori obvious, and indeed in other resource theories such as the resource theory of magic states, there is a strict inequality between two analogous robustness measures \cite{Reg17}. This demonstrates a rather curious property of $k$-coherence, where the quantification of several measures of this resource reduces to a single quantity for all pure states.

Similar work concerning measures of the entanglement of a quantum state, rather than its coherence, has previously been carried out \cite{Rud01,Rud05,VT99}, and a similar conjecture about the equivalence of two measures of Schmidt rank $k$-entanglement on pure states was made in \cite{Nconj}. As an application of our main result, we also prove this conjecture---we show that these two different robustnesses of Schmidt rank $k$-entanglement agree with each other when restricted to pure states, establishing a computable formula which allows for the quantification of pure-state $k$-entanglement and significantly improves on bounds for the robustnesses known in the literature previously \cite{Clarisse}.

The paper is organized as follows. In Section~II, we introduce the mathematical preliminaries required to discuss these coherence measures properly. In Section~III, we present our main result: that the standard robustness of $k$-coherence  and the generalized robustness  of $k$-coherence are equal for all pure states. Our methodology involves the use of semidefinite programming, with the bulk of the ``heavy lifting'' done via a separating hyperplane argument. Much like in \cite{CGJLP, JLP}, the formula that we find ``branches'' into one of $n$ different formulas, depending on how close the entries of the pure state vector are to each other (though the techniques used here are different than in those papers). Section~IV is dedicated to  discussions of how to explicitly construct the closest $k$-incoherent state in certain special cases, since the proof of our main result is non-constructive. In Section~V, we detail the connection between robustness of $k$-coherence with robustness of $k$-entanglement and show why our results answer the corresponding question about Schmidt rank $k$-entanglement. We end with conclusions in Section~VI.

\section{II. Preliminaries}

We use kets like $\ket{v}$ to denote unit vectors (pure states) in $\mathbb{C}^n$, lowercase Greek letters like $\rho$ to denote arbitrary density matrices (positive semidefinite matrices with trace $1$), outer products like $\ketbra{v}{v}$ to specifically denote pure state density matrices, and $\cD_n$ to denote the set of $n \times n$ density matrices (pure or mixed). We will use the notation  $\mathbf{x}$ to denote an unnormalized complex vector,   $\mathbf{1} = (1,\ldots,1)^t$ the all-ones vector,  and $A\succeq 0$ to mean that the matrix $A$ is positive semidefinite. We say that a pure state $\ket{v}$ is $k$-incoherent if it has $k$ or fewer non-zero entries (when written in the standard computational basis $\{\ket{1}, \dots, \ket{n}\}$ of $\mathbb{C}^n$), and we say that a density matrix is $k$-incoherent if it is in the set
\begin{align*}
    \cI_k\defeq \Big\{\sum_i p_i\ketbra{v_i}{v_i}\,: & \ p_i\geq 0, \sum_i p_i=1, \\
    & \ \ket{v_i} \textnormal{ is $k$-incoherent } \forall i \ \Big\}.
\end{align*}
For convenience, we define $\cI = \cI_1$, which is equal to the set of diagonal density matrices, and we also notice that $\cI_n = \cD_n$ (the set of \emph{all} density matrices). For intermediate values of $k$, it is not difficult to show that these sets are convex and satisfy $\cI_k\subsetneq \cI_{k+1}$, so they form a hierarchy that interpolates between $\cI$ and $\cD_n$ in a fairly natural way.

Analogous to the setting of entanglement \cite{VT99}, the robustness of coherence \cite{NBCPJA16} of a given state $\rho\in \cD_n$ is defined by 
\begin{eqnarray}\label{eq:robust}
    R(\rho) & \defeq & \min_{\tau \in \cD_n} \left\{ s \geq 0 \, \Big| \, \frac{\rho + s\tau}{1 + s} \in \mathcal{I}\right\}.
\end{eqnarray}
More generally, for $k \in \{2, 3, \dots, n\}$ one can define two different robustnesses of $k$-coherence---the ``standard'' and ``generalized'' robustness of $k$-coherence, respectively---as follows \cite{RBCLAWFP17,Reg17}:
\begin{eqnarray}
    R_k^s(\rho)& \defeq &\min_{\sigma\in \cI_k}\left\{s\geq 0\,:\,\frac{\rho+s\sigma}{1+s}\in \cI_k\right\} \quad \text{and}\label{eq:standard_k_coh}\\
    R_k^g(\rho)& \defeq &\min_{\tau\in \cD_n}\left\{s\geq 0\,:\, \frac{\rho+s\tau}{1+s}\in \cI_k\right\}.\label{eq:gen_k_coh}
\end{eqnarray}
We note that the measure~\eqref{eq:standard_k_coh} only makes sense if $k \geq 2$, since if $k = 1$ then $\cI_1$ does not span the space of density matrices, so if $\rho$ is not diagonal then we cannot find \emph{any} value of $s$ such that $(\rho + s\sigma)/(1+s) \in \cI_1$. On the other hand, the measure~\eqref{eq:gen_k_coh} indeed reduces to exactly the usual robustness of coherence~\eqref{eq:robust} when $k = 1$.

Since $\cI_k\subsetneq\cD_n$ for fixed $n$ and $k=1,\dots, n-1$, we have $R_k^g\leq R_k^s$, and numerics can be used to straightforwardly show that in fact $R_k^g(\rho) < R_k^s(\rho)$ for most randomly-chosen $\rho \in \cD_n$ (both of these measures of $k$-coherence can be computed numerically via semidefinite programming). However, we will show that these two $k$-coherence measures do in fact coincide when restricted to pure states. That is, $R_k^g(\ketbra{v}{v}) = R_k^s(\ketbra{v}{v})$ for all $\ket{v} \in \mathbb{C}^n$.

\section{III. Main result}

We are now ready to state our main result: a formula for $R_k^s(\ketbra{v}{v})$ that agrees with the formula for $R_k^g(\ketbra{v}{v})$ that was derived in \cite{Reg17} and thus shows that these two measures of coherence do indeed coincide on pure states. Note that the theorem is stated only for pure states $\ket{v}=(v_1, \dots, v_n)^t$ with real entries satisfying $v_1 \ge \cdots \ge v_n\ge 0$. This is not actually a restriction, since if $\ket{v}$ is not of this form then it can be converted to this form via a diagonal unitary and/or a permutation matrix, and these operations do not affect the value of $R_k^s$ or $R_k^g$.

For convenience, for each $1 \leq j \leq n$ we define $s_j := \sum_{i=j}^n v_i$. Then our main result is as follows.

\begin{theorem}\label{thm:main}
    Let $\ket{v}=(v_1, \dots, v_n)^t$ be a pure state with $v_1 \geq v_2 \geq \cdots \geq v_n \geq 0$. Fix $k \in \{2, 3, \dots, n\}$ and let $\ell \in \{2,3,\ldots,k\}$ be the largest integer such that $v_{\ell-1} \geq s_\ell/(k-\ell+1)$ (set $\ell = 1$ if no such integer exists). Then
    \begin{align*}
        R_k^s(\ketbra{v}{v}) = R_k^g(\ketbra{v}{v}) = \frac{s_\ell^2}{k-\ell+1} - \sum_{i=\ell}^{n}v_i^2.
    \end{align*}
\end{theorem}

Before we proceed to prove the above result, let us provide some intuition for the formula appearing in Theorem \ref{thm:main}. This expression can be related to the   $k$-support norm $\knorm{\ket{v}}{k}$, defined first in \cite{Ksupport} and shown to correspond to $R_k^g$ for pure states in \cite[Theorem~10]{Reg17}, in the sense that $R_k^g(\ketbra{v}{v}) = \knorm{\ket{v}}{k}^2-1$. As a valid norm on $\mathbb{C}^n$, the $k$-support norm can be alternatively computed using its dual norm, given for any $\mathbf{x} \in \mathbb{C}^n$ by \cite{Ksupport,RPCBSA}
\begin{equation}\begin{aligned}
    \knormd{\mathbf{x}}{k} &=  \max \left\{ \left|\mathbf{x}^\dagger \ket{v}\right| \,:\, \ket{v} \textnormal{ is $k$-incoherent } \right\}\\
    &= \sqrt{\sum_{i=1}^{k} x^\downarrow_i}
\end{aligned}\end{equation}
where $x^\downarrow_i$ denotes the coefficients of $\mathbf{x}$ arranged so that $|x^\downarrow_i| \geq \cdots \geq |x^\downarrow_n|$. By norm duality, we then have
\begin{equation}\begin{aligned}\label{eq:ksupp_dual}
    \knorm{\mathbf{x}}{k} = \max \left\{ \left|\mathbf{x}^\dagger \mathbf{a}\right| \,:\, \knormd{\mathbf{a}}{k} \leq 1 \right\}.
\end{aligned}\end{equation}
We further remark that $\knorm{\mathbf{x}}{1} = \sum_i |x_i|$ and $\knorm{\mathbf{x}}{n} = \sqrt{\mathbf{x}^\dagger\mathbf{x}}$, and hence the $k$-support norm can be seen as a natural way to interpolate between the $\ell_1$ and $\ell_2$ norms.

The remainder of this section is devoted to proving the above result. Although the lower bound follows already from \cite{Reg17}, we rederive it directly here, as it is no extra work to do so. We will separately show that the claimed formula is a lower bound and an upper bound of the robustnesses.

\subsection{Lower bound via semidefinite programming duality}

To show that the formula described by Theorem~\ref{thm:main} is a lower bound of the robustnesses of $k$-coherence, we use semidefinite programming duality techniques. The minimization problems~\eqref{eq:standard_k_coh} and~\eqref{eq:gen_k_coh} that define $R_k^s$ and $R_k^g$ are (primal) semidefinite programs, and their dual programs can be written in the following forms:
\begin{eqnarray}
    R_k^s(\rho) & = & \max_{W\in \cI_k^o}\big\{\tr(\rho W)\,:\,I-W\in \cI_k^{\circ}\big\}-1\label{eq:rob_standard_dual} \\
    R_k^g(\rho) & = & \max_{W \succeq 0}\big\{\tr(\rho W)\,:\, I-W\in \cI_k^{\circ}\big\}-1,\label{eq:rob_gen_dual}
\end{eqnarray}
where $\cI_k^{\circ}$ is the dual cone of $\cI_k$, defined by \cite{RBCLAWFP17}
\begin{align*}
    \cI_k^{\circ} & \defeq \{W = W^\dagger\,:\,\tr(W\rho)\geq 0 \ \ \forall \, \rho\in \cI_k\} \\
    & = \{W = W^\dagger\,:\,\textnormal{all } k\times k \textnormal{ principal submatrices } \\
    & \qquad \qquad \qquad \quad  W[i_1, \dots, i_k]\textnormal{ of } W \textnormal{ are }\succeq 0\},
\end{align*}
where $W[i_1, \dots, i_k]$ denotes the principal submatrix of $W$ containing rows and columns $i_1,\ldots,i_k$.
Note that since $R_k^g\leq R_k^s$, we just want to find a particular $W$ that attains the optimal value claimed by the theorem in the maximization problem~\eqref{eq:rob_gen_dual}, since that establishes a lower bound of $R_k^g$ and thus of $R_k^s$ as well.

To establish this formula as a lower bound, we first define (for convenience of notation) the quantities
\begin{align*}
    \alpha := \frac{s_\ell}{k-\ell+1} \quad \text{and} \quad \beta := \sqrt{\alpha s_\ell + \sum_{j=1}^{\ell-1}v_j^2}.
\end{align*}
In particular, this definition gives $v_{\ell-1} \ge \alpha$ and
\begin{align*}
    \beta^2 - 1 = \alpha s_\ell + \sum_{j=1}^{\ell-1}v_j^2 - 1 = \frac{s_\ell^2}{k-\ell+1} - \sum_{i=\ell}^{n}v_i^2,
\end{align*}
which is our claimed formula for $R_k^s(\ketbra{v}{v})$ and $R_k^g(\ketbra{v}{v})$.

In the case when $\ell \geq 2$, we set
\begin{equation}\label{eq:feasible_a}
    \mathbf{a} := \frac{1}{\beta}\big(v_1,v_2,\ldots,v_{\ell-1},\alpha,\alpha,\ldots,\alpha\big)^t \in \mathbb{R}^n
\end{equation}
and then define
\[
    W := \mathbf{a}\mathbf{a}^t.
\]
We claim that $W$ is a feasible point of the semidefinite program~\eqref{eq:rob_gen_dual} that produces the desired objective value. To see this, we first note that $W \succeq 0$ and
\begin{align*}
    \tr(\ketbra{v}{v}W) - 1 & = \frac{1}{\beta^2}\left(\sum_{j=1}^{\ell-1} v_j^2 + \sum_{j=\ell}^n v_j\alpha\right)^2 - 1 \\
    & = \frac{1}{\beta^2}\left(\alpha s_\ell + \sum_{j=1}^{\ell-1} v_j^2\right)^2 - 1 \\
    & = \beta^2 - 1,
\end{align*}
as desired.

All that remains is to show that $I - W \in \cI_k^{\circ}$, which is equivalent to the statement that $I_k \succeq W[i_1, \dots, i_k]$ for all $i_1 < \cdots < i_k$. To see why this is the case, note that $W[i_1, \dots, i_k]$ has rank one and its maximum eigenvalue is thus $\tr(W[i_1, \dots, i_k])$. Because $v_1 \ge \cdots \ge v_{\ell-1} \ge \alpha$, the $k \times k$ principal submatrix of $W$ with the largest trace is $W[1, \dots, k]$, which has trace
\begin{align*}
    \tr(W[1, \dots, k]) & = \frac{1}{\beta^2}\left(\sum_{j=1}^{\ell-1}v_j^2 + (k-\ell+1) \alpha^2\right) \\
    & = \frac{1}{\beta^2}\left(\sum_{j=1}^{\ell-1}v_j^2 + \alpha s_\ell\right) \\
    & = 1.
\end{align*}
Thus it is indeed the case that $I_k \succeq W[i_1, \dots, i_k]$, which shows that $W$ is a feasible point of the semidefinite program and completes the proof of the lower bound when $\ell \geq 2$.

To see this the desired lower bound also holds when $\ell = 1$, we set $W = J_n/k$, where $J_n$ is the $n \times n$ all-ones matrix. It is straightforward to check that $W \succeq 0$ and that every $k \times k$ principal submatrix of $W$ has largest eigenvalue $1$, so $I - W \in \cI_k^{\circ}$. Thus $W$ is a feasible point of the semidefinite program~\eqref{eq:rob_gen_dual} with objective value
\begin{align*}
    \tr(\ketbra{v}{v}W) - 1 & = \frac{1}{k}\left(\sum_{j=1}^{n} v_j\right)^2 - 1 = \beta^2 - 1,
\end{align*}
as desired, which completes the proof of the $\ell = 1$ part of the lower bound.

We have now shown that the quantity described by the theorem is indeed a lower bound on both robustnesses:
\begin{align*}
    \frac{s_\ell^2}{k-\ell+1} - \sum_{i=\ell}^{n}v_i^2 \leq R_k^g(\ketbra{v}{v}) \leq R_k^s(\ketbra{v}{v}).
\end{align*}

Notice that $\mathbf{a}$ in Eq.~\eqref{eq:feasible_a} is precisely the vector which achieves the maximum in the dual expression for the $k$-support norm $\knorm{\ket{v}}{k}$ in Eq.~\eqref{eq:ksupp_dual}.

\subsection{The upper bound}

To see that this quantity is also an upper bound, we return to the original formulation of these robustnesses~\eqref{eq:standard_k_coh} and~\eqref{eq:gen_k_coh} as minimization problems. We want to show that there exists $\sigma \in \cI_k$ such that
\[
    \frac{\ketbra{v}{v} + s\sigma}{1 + s} \in \cI_k, \quad \text{where} \quad s = \beta^2 - 1,
\]
since that would then imply $R_k^g(\ketbra{v}{v}) \leq R_k^s(\ketbra{v}{v}) \leq \beta^2 - 1$.

First, we need the following simple lemma.

\begin{lemma}\label{lem:inclusion_in_i2}
    Let $S \in M_n$ be a Hermitian matrix with non-negative diagonal entries, non-positive off-diagonal entries, and non-negative row sums. Then $S$ is a non-negative sum of matrices in $\cI_2$.
\end{lemma}

\begin{proof}
    Without loss of generality, we can consider the case when the row sums of $S$ equal $0$, since if they are strictly positive then we can just write $S = \tilde{S} + D$ where $\tilde{S}$ has $0$ row sums and $D \in \cI_1 \subset \cI_2$.
    
    For each $1 \le i < j \le n$, let $s_{i,j}$ be the $(i,j)$-entry of $S$ and let $G^{(i,j)}$ be the $n\times n$ matrix with entries defined by
    \begin{align*}
        (G^{(i,j)})_{k,\ell} & = \begin{cases}
            s_{i,j}, & \text{if} \ \{k,\ell\} = \{i,j\} \\
            |s_{i,j}|, & \text{if} \ \{k,\ell\} = \{i\} \ \text{or} \ \{k,\ell\} = \{j\} \\
            0, & \text{otherwise}.
        \end{cases}
    \end{align*}

    Observe that each $G^{(i,j)}$ is positive semidefinite because it is diagonally dominant, and is only non-zero on a $2 \times 2$ submatrix and is thus a non-negative sum of matrices in $\cI_2$. Since $S = \sum_{i=1}^n\sum_{j=i+1}^n G^{(i,j)}$, it follows that $S$ is a non-negative sum of matrices in $\cI_2$ too.
\end{proof}

We now divide the proof into three cases: (1)~$\ell=k$, (2)~$\ell=1$, and (3)~$1<\ell<k$. Case~(1) is the ``easy'' case, and for it we are able to construct an explicit solution $s\sigma$. On the other hand, cases~(2) and~(3) are much more involved and we are only able to prove existence of a solution---we do not explicitly construct it.

\subsubsection{The solution when \texorpdfstring{$\ell = k$}{l=k}} 

We start by defining
\begin{align*}
    \mathbf{u}_j := \sum_{i=1}^{k-1} \sqrt{\frac{v_j}{s_k}}v_i \ket{i} + \sqrt{s_k v_j} \ket{j} \quad \text{for} \quad j = k, \dots, n,
\end{align*}
where $\ket{i}$ refers to the $i$-th standard basis vector. We then define $s\sigma := \sum_{j=k}^n \mathbf{u}_j\mathbf{u}_j^t - \ketbra{v}{v}$. Since each $\mathbf{u}_j$ has at most $k$ non-zero entries, it follows that $\ketbra{v}{v} + s\sigma \in (1+s)\cI_k$. Also, direct calculation shows that $s\sigma = (O_{k-1} \oplus S)$, where
\[
    S = \diag(v_ks_k, v_{k+1}s_{k}, \ldots, v_ns_k)-(v_iv_j)_{k \le i, j \le n},
\]
which has non-negative diagonal entries and non-positive off-diagonal entries. Lemma~\ref{lem:inclusion_in_i2} thus tells us that $S$ is a non-negative sum of matrices in $\cI_2 \subset \cI_k$, so $s\sigma$ is as well. Moreover,
\[
    \tr(S) = \sum_{i=k}^{n} v_is_k - \sum_{i=k}^n v_i^2 = s_k^2 - \sum_{i=k}^n v_i^2 = \beta^2 - 1,
\]
which is the desired objective value. This completes the proof of the $\ell = k$ case.


\subsubsection{The solution when \texorpdfstring{$\ell=1$}{l=1}}

Notice that in this case we have $v_2 < s_2/(k-1)$, which is equivalent to $v_1 < s_1/k$. Let $\cS_n$ be the set of symmetric matrices with non-negative diagonal entries, non-positive off-diagonal entries, and row sums equal to zero. Let $\cT_k$ be the convex
hull of the matrices of the form $\mathbf{x}\mathbf{x}^t$, where $\mathbf{x} \in \IR^n$ has $k$ nonzero entries each equal to $s_1/k$. 
Our goal now is to show that there exists a matrix $S := s\sigma \in \cS_n$ such that $\tau := \ketbra{v}{v} + S \in \cT_k$. If such an $S$ does indeed exist then we are done, since Lemma~\ref{lem:inclusion_in_i2} then tells us that $S$ is a non-negative sum of matrices in $\cI_2$, so $\sigma \in \cI_2 \subset \cI_k$, $\tau \in \cI_k$ by construction since each $\mathbf{x}$ has $k$ non-zero entries, and $\tr(S) = s = s_1^2/k - 1$ (as desired) since $\tau \in \cT_k$ implies $\tr(\tau) = s_1^2/k$. Thus $S$ defines a feasible point of the minimization problem~\eqref{eq:standard_k_coh} that produces the desired value in the objective function.

To see that there exists $S \in \cS_n$ such that $\tau \in \cT_k$, suppose for the sake of contradiction that no such $S$ exists. It is straightforward to verify that $\cS_n$ is a closed convex cone consisting of non-negative combinations of matrices of the form $(\ket{i}-\ket{j})(\bra{i}-\bra{j})$. Thus, since no such $S \in \cS_n$ exists, there is a separating hyperplane on real symmetric matrices separating the convex cone $\ketbra{v}{v} + \cS_n$ from the compact convex set $\cT_k$. This separating hyperplane may be represented by a real symmetric matrix $H$ with the property that
\[
    \tr\big((\ketbra{v}{v}+A)H\big) > \tr(\tau H)
\]
for all $A \in \cS_n$ and $\tau \in \cT_k$. By convexity of the sets in question, this is equivalent to
\[
    \tr\big((\ketbra{v}{v}+c(\ket{i}-\ket{j})(\bra{i}-\bra{j}))H\big) > \tr (\mathbf{x}\mathbf{x}^t H)
\]
for all $c > 0$, all standard basis states $\ket{i}$, $\ket{j}$, and all $\mathbf{x}\in \IR^n$ with $k$
nonzero entries, each equal to $s_1/k$. Dividing both sides by $s_1^2$ then gives us the following equivalent condition, where $\mathbf{u} := (1/s_1)\ket{v}$:
\begin{align}\label{eq:sep_hyp}
    \tr\big((\mathbf{u}\mathbf{u}^t+c(\ket{i}-\ket{j})(\bra{i}-\bra{j}))H\big) > \tr (\mathbf{x}\mathbf{x}^t H)
\end{align}
for all $c > 0$, all standard basis states $\ket{i}$, $\ket{j}$, and all $\mathbf{x}\in \IR^n$ with $k$
nonzero entries, each equal to $1/k$.

We will now show that there does not exist a real symmetric matrix $H$ satisfying this condition, so our original assumption that there is no $S \in \cS_n$ with $\tau \in \cT_k$ must be incorrect. We break down the proof of the non-existence of $H$ into two lemmas, and throughout them we let $V(n,k)$ denote the set of all $\mathbf{x}\in \IR^n$ with $k$ nonzero entries each equal to $1/k$.

\begin{lemma} \label{thm-1}
    Let $\mathbf{u} = (u_1, \dots, u_n)^t \in \IR^n$ with $u_i \in [0, 1/k]$ and $u_1 + \cdots + u_n = 1$. The following statements are equivalent.
    \begin{itemize}
        \item[{\rm (a)}] There is a symmetric matrix $H = (h_{i,j})$ such that
        \[
            \tr[(\mathbf{u}\mathbf{u}^t +c(\mathbf{w}\mathbf{w}^t))H] > \tr(\mathbf{x}\mathbf{x}^t H)
        \]
        for all $c \ge 0$, $\mathbf{w} = \ket{i}-\ket{j}$, and $\mathbf{x}\in V(n,k)$.
        
        \item[{\rm (b)}] There is a symmetric matrix $H = (h_{i,j})$ (with positive entries) such that
        
        \medskip\centerline{$h_{i,i} + h_{j,j} \ge h_{i,j}+h_{j,i}$ for all $1 \le i < j \le n$, and}
        
        \medskip\centerline{$\tr(\mathbf{u}\mathbf{u}^t H) > \tr(\mathbf{x}\mathbf{x}^t H)$ for all $\mathbf{x}\in V(n,k)$.}
    \end{itemize}
\end{lemma}

\begin{proof}
    Suppose (a) holds. If there exist $1\le i<j\le n$ such that $h_{i,i}+h_{j,j} - h_{i,j}-h_{j,i}< 0$, then letting $\mathbf{w} =\ket{i}-\ket{j}$, we can find a sufficiently large $c > 0$ such that for all $\mathbf{x} \in \IR^n$ with $k$ nonzero entries each equal to $1/k$,  
    $$\tr[(\mathbf{u}\mathbf{u}^t+c(\mathbf{w}\mathbf{w}^t))H] < \tr (\mathbf{x}\mathbf{x}^t H),$$
    for all $\mathbf{x}\in V(n,k)$, which is a contradiction. Therefore $H = (h_{i,j})$ such that $h_{i,i}+h_{j,j} - h_{i,j}-h_{j,i}\ge 0$ for all $i < j$. Now,
    $$\tr[(\mathbf{u}\mathbf{u}^t+c(\mathbf{w}\mathbf{w}^t))H] > \tr (\mathbf{x}\mathbf{x}^t H)$$
    for all $c \ge 0$. We conclude that
    $$\tr(\mathbf{u}\mathbf{u}^t H) > \tr (\mathbf{x}\mathbf{x}^t H).$$

    In the opposite direction, now suppose (b) holds. Then $\tr[(\mathbf{w}\mathbf{w}^t H)] \ge 0$, and hence, for all $c \ge 0$, we have
    $$\tr[(\mathbf{u}\mathbf{u}^t+c (\mathbf{w}\mathbf{w}^t))H]\ge \tr(\mathbf{u}\mathbf{u}^t H) > \tr (\mathbf{x}\mathbf{x}^t H),$$
    for all $\mathbf{x}\in V(n,k)$. To show that we may assume $H$ has positive entries, let $J_n$ be the all-ones matrix. Then $\tr(\mathbf{u}\mathbf{u}^t J_n) = 1 = \tr(\mathbf{x}\mathbf{x}^t J_n)$. Thus, we may replace $H$ by $H+\mu J_n$ for sufficiently large $\mu > 0$ so that the resulting matrix has positive entries without changing the conditions $h_{i,i} + h_{j,j} \ge h_{i,j} + h_{j,i}$ for $1 \le i < j \le n$, and $\tr(\mathbf{u}\mathbf{u}^t H) > \tr (\mathbf{x}\mathbf{x}^t H)$ for all $\mathbf{x}\in V(n,k)$.
\end{proof}

\begin{lemma}\label{thm2}
    There is no matrix $H$ satisfying condition~(b) of Lemma~\ref{thm-1}.
\end{lemma}

\begin{proof}
    Let $B$ be the set of all $k$-tuples $\mathbf{b} = (b_1, \dots, b_k)$ with $b_1, \dots, b_k\in \{1, \dots, n\}$. Suppose such a positive symmetric matrix $H$ exists. We may replace $H$ by $\mu H$ for some $\mu > 0$ and assume that 
    $$\max\{ \mathbf{x}_{\mathbf{b}}^t H \mathbf{x}_\mathbf{b}: \mathbf{b}\in B\} = 1,$$
    where $\mathbf{x}_\mathbf{b} = \frac{1}{k} \sum_{j=1}^k \ket{b_j}$.
    Consider the set of vectors
    \begin{eqnarray*}
       & \cW = \{ \mathbf{w} = (w_1, \dots, w_n)^t:  w_i \in [0, 1/k], \\
       &\quad  \quad w_1 + \cdots + w_n = 1 
    < \mathbf{w}^t H \mathbf{w}\}.
    \end{eqnarray*}
    Then $\mathbf{u}$ from Theorem \ref{thm-1} (b) is in $\cW$, and there is a vector $\mathbf{w} \in \cW$ attaining the maximum values $\mathbf{w}^tH\mathbf{w}$. Should there be more than one vector 
    $\mathbf{w}= (w_1, \dots, w_n)^t \in \cW$ attaining the value, we consider any $\mathbf{w}$ with maximum number of entries equal to $1/k$.

    Our goal is to show that we can choose $\mathbf{w}$ such that $k$ of its entries are $1/k$. To this end, suppose for now that there exist $1 \le r < s \le n$ such that ${w}_r, {w}_s \in (0,1/k)$. Then we can choose a suitable $\lambda\neq0$ and replace $\mathbf{w}$ by $\tilde{\mathbf{w}} = \mathbf{w} + \lambda(\ket{r} - \ket{s}) \in \cW$ so that $\{\tilde{w}_r, \tilde  {w}_s\} = \{ {w}_r+ {w}_s, 0\}$ or $\{\tilde  {w}_r, \tilde  {w}_s\} = \{1/k,  {w}_r+ {w}_s - 1/k\}$ depending on whether $ {w}_r+ w_s \ge 1/k$ or not. 

    We now claim that $\tr(\tilde \w\tilde \w^t H) = \tr(\w\w^tH)$. To see why this is the case, we write $H$ in terms of its columns: $H = [ \ \mathbf{h}_1 \ | \ \mathbf{h}_2 \ | \ \cdots \ | \ \mathbf{h}_n \ ]$, and compute
    \begin{eqnarray*}
        \tr (\tilde \w\tilde \w^t H) &=& \tr (\w \w^t H) + 2\lambda (\w^t \mathbf{h}_r - \w^t \mathbf{h}_s) \\
        &&\quad + \ \lambda^2 (h_{r,r} + h_{s,s} - h_{r,s}-h_{s,r}).
    \end{eqnarray*}
    This quantity cannot be strictly greater than $\tr(\w \w^t H)$, since that would contradict the fact that $\w$ was chosen to maximize $\w^tH\w$. On the other hand, if it is strictly less than $\tr(\w\w^t H)$ then
    \[
        2\lambda (\w^t \mathbf{h}_r - \w^t \mathbf{h}_s) + \lambda^2 (h_{r,r} + h_{s,s} - h_{r,s}-h_{s,r}) < 0.
    \]
    Replacing $\lambda$ by $-\lambda$ would change the sign of this inequality, so we could construct another vector $\hat{\w}$ in a manner similar to $\tilde{\w}$ by replacing $\lambda$ by $-\lambda$ (and possibly making it smaller, if necessary, so that the entries are still between $0$ and $1/k$) so that $\tr (\hat \w\hat\w^t H) > \tr(\w\w^t H)$, which again contradicts the fact that $\w$ was chosen to maximize $\w^tH\w$. Thus it must be the case that $\tr(\tilde \w\tilde \w^t H) = \tr(\w\w^tH)$.
    
    By the assumption on $\mathbf{w}$, we cannot have $\{w_r, w_s\} = \{1/k, w_r+w_s - 1/k\}$. But then we can repeat the arguments until we get more and more entries equal to $0$. Ultimately, we will obtain two entries $w_r, w_s$ with $w_r+w_s \ge 1/k$ and arrive at the situation $\{w_r, w_s\} = \{1/k, w_r+w_s - 1/k\}$. 
    
    It follows that, to avoid this condition, there cannot be two entries of $\mathbf{w}$ lying in $(0,1/k)$ to begin with. But then $\mathbf{w}$ must have exactly $k$ nonzero entries each equal to $1/k$. Thus $\mathbf{w} = \mathbf{x}_\mathbf{b}$ for some $\mathbf{b} \in B$ and 
    $$\mathbf{w}^tH\mathbf{w} = \mathbf{u}^tH\mathbf{u} > 1 =  \max\{\mathbf{x}_{\mathbf{b}}^t H \mathbf{x}_\mathbf{b}: \mathbf{b}\in B\},$$
    which is a contradiction that shows that $H$ does not exist.
\end{proof}

Combining Lemmas~\ref{thm-1} and~\ref{thm2} immediately shows that there does not exist a real symmetric matrix $H$ satisfying Inequality~\eqref{eq:sep_hyp}, which is the contradiction we have been striving for that completes the proof in this case.

\subsubsection{The solution when \texorpdfstring{$1 < \ell < k$}{1<l<k}}

Finally, we now consider the case when $1 < \ell < k$. Recall that our aim is to show there exists a $\sigma \in \cI_k$ such that $\tau := \ketbra{v}{v} + s\sigma \in (1+s)\cI_k$, where $s = \beta^2-1$ is the quantity described in the statement of the theorem. We will find that we can choose $\sigma$ to have the block diagonal form $s\sigma = O_{\ell-1} \oplus S$ with $S \in \cS_{n-\ell+1}$ (and $\cS_{n-\ell+1}$ is as defined in the $\ell = 1$ case).

First, let $\tilde{\mathbf{v}} = (v_\ell, \dots, v_n)^t$. By mimicking the argument presented for the $\ell = 1$ case, we know that there exists $S \in \cS_{n-\ell+1}$ such that $\tilde{\tau} := \tilde{\mathbf{v}}\tilde{\mathbf{v}}^t + S \in \cT_{k-\ell+1}$ (where $\cT_{k-\ell+1}$ is now the convex hull of the matrices of the form $\tilde{\mathbf{x}}\tilde{\mathbf{x}}^t$, where $\tilde{\mathbf{x}} \in \IR^{n-\ell+1}$ has $k-\ell+1$ nonzero entries each equal to $s_\ell/(k-\ell+1)$). Thus we can write
$$
    \tilde{\tau} = \sum p_{\tilde{\mathbf{x}}}\tilde{\mathbf{x}}\tilde{\mathbf{x}}^t,
$$
where the $p_{\tilde{\mathbf{x}}}$'s are coefficients in the convex combination.

For each $\tilde{\mathbf{x}}$, construct $\mathbf{x} \in \IR^n$ via $\mathbf{x} = (v_1, \dots, v_{\ell-1})^t \oplus \tilde{\mathbf{x}}$, and then define
$$
    \tau := \sum p_{\tilde{\mathbf{x}}}\mathbf{x}\mathbf{x}^t.
$$
Notice that $\tau \in (1+s)\cI_k$ by construction and
\[
    \tr(S) = \tr(\tilde{\tau} - \tilde{\mathbf{v}}\tilde{\mathbf{v}}^t) = \frac{s_\ell^2}{k-\ell+1} - \sum_{i=\ell}^{n}v_i^2,
\]
which is the desired objective value. Thus if we can show that $\tau = \ketbra{v}{v} + (O_{\ell-1} \oplus S)$ then we are done.

By construction of $\tau$, these two matrices trivially agree on their bottom-right $(n-\ell+1)\times(n-\ell+1)$ submatrices. The top-left $(\ell-1)\times(\ell-1)$ submatrix of $\tau$ equals
\begin{align*}
    & \sum p_{\tilde{\mathbf{x}}}(x_1,\dots, x_{\ell-1})^t(x_1, \dots, x_{\ell-1}) \\
    & \qquad \qquad \qquad \qquad = (v_1,\dots, v_{\ell-1})^t(v_1, \dots, v_{\ell-1}),
\end{align*}
which is the top-left $(\ell-1)\times(\ell-1)$ submatrix of $\ketbra{v}{v} + (O_{\ell-1} \oplus S)$.

Finally, for the bottom-left $(n-\ell+1) \times (\ell-1)$ submatrices, first recall that $S$ has row sums $0$, so the row sums of $\tilde{\tau} = \tilde{\mathbf{v}}\tilde{\mathbf{v}}^t + S$ are $v_\ell s_\ell, v_{\ell+1}s_\ell, \ldots, v_n s_\ell$. Thus
\begin{align*}
    \tilde{\tau}{\bf 1}_{n-\ell+1} & = s_\ell(v_{\ell},v_{\ell+1},\ldots,v_n)^t.
\end{align*}
On the other hand, we can directly compute
\begin{align*}
    \tilde{\tau}{\bf 1}_{n-\ell+1} & = \sum p_{\tilde{\mathbf{x}}}\tilde{\mathbf{x}}\tilde{\mathbf{x}}^t{\bf 1}_{n-\ell+1} = s_\ell\sum p_{\tilde{\mathbf{x}}}\tilde{\mathbf{x}}.
\end{align*}
By combining the two above formulas for $\tilde{\tau}{\bf 1}_{n-\ell+1}$, we see that
$$
    \sum p_{\tilde{\mathbf{x}}}\tilde{\mathbf{x}} = (v_{\ell},v_{\ell+1},\ldots,v_n)^t.
$$
It follows that the bottom-left $(n-\ell+1) \times (\ell-1)$ submatrix of $\tau$ equals
\begin{align*}
    & \sum p_{\tilde{\mathbf{x}}}\tilde{\mathbf{x}}(v_1, v_2, \ldots, v_{\ell-1}) \\
    & \qquad \qquad \qquad = (v_\ell, v_{\ell+1}, \ldots, v_n)^t(v_1, v_2, \ldots, v_{\ell-1}),
\end{align*}
which is the bottom-left $(n-\ell+1) \times (\ell-1)$ submatrix of $\ketbra{v}{v} + (O_{\ell-1} \oplus S)$, as desired. We thus conclude that $\tau = \ketbra{v}{v} + (O_{\ell-1} \oplus S)$, which completes the proof.

\section{IV. Implementation}

Despite finding a formula for $R_k^s(\ketbra{v}{v})$, we did not explicitly construct the optimal matrix $\sigma$ in the defining minimization problems~\eqref{eq:standard_k_coh} and~\eqref{eq:gen_k_coh} except in the $\ell = k$ case. Rather, we gave an existence proof based on the non-existence of a hyperplane separating two convex sets.

However, we did unearth some of the structure of an optimal $\sigma$ that makes it easier to find than na\"ive optimization methods. In particular, although the constraint $S \in \cI_k$ can be implemented using semidefinite programming, the complexity of this constraint grows combinatorially with the matrix dimension $n$, making the computation infeasible in practice beyond low-dimensional cases. However, the proof of Theorem~\ref{thm:main} showed that the optimal $\sigma$ can be chosen to have the form $O_{\ell-1} \oplus S$, where instead of requiring $S \in \cI_k$, we require it to have non-negative diagonal entries, non-positive off-diagonal entries, and row sums $0$, which can be checked via \emph{linear} programming.

Similarly, instead of requiring $\ketbra{v}{v} + s\sigma$ to be in $(1+s)\cI_k$, we can require it to be a convex combination of the (finitely many) matrices of the form $\mathbf{x}\mathbf{x}^t$, where each $\mathbf{x}$ is of the form $\mathbf{x} = (v_1,v_2,\ldots,v_{\ell-1}) \oplus \tilde{\mathbf{x}}$ for some $\tilde{\mathbf{x}}$ with exactly $k-\ell+1$ non-zero entries, each equal to $s_\ell/(k-\ell+1)$. Again, linear programming can be used to find the coefficients in this convex combination, so we can construct $\sigma$ via linear programming, which is significantly faster than semidefinite programming in practice. MATLAB code that implements this linear program (as well as the na\"{i}ve methods based on semidefinite programming via CVX \cite{CVX}, and the result of Theorem~\ref{thm:main}) is available online from \cite{MATLABcode}.

\section{V. Application to entanglement measures}

Our result can be extended from the robustnesses of $k$-coherence to the analogous measures of entanglement of pure states. Let $SR(\ket{v})$ denote the Schmidt rank of the pure state $\ket{v}$ and let $SN(\rho)$ denote the Schmidt number \cite{TH00} of a mixed state $\rho \in \cD_{mn}$. That is, $SN(\rho)$ is the least integer $k$ such that we can write
\[
    \rho = \sum_i p_i \ketbra{v_i}{v_i}
\]
with $p_i \geq 0$ and $SR(\ket{v_i}) \leq k$ for all $i$. The $k$-projective tensor norm \cite{Rud00,Nconj} and the $k$-robustnesses of entanglement \cite{Clarisse,Rud05,VT99} are defined, respectively, via
\begin{eqnarray}
    \big\|X\big\|_{\gamma,k} & \defeq & \inf\Big\{ \sum_i |c_i| : X = \sum_i c_i\ketbra{v_i}{w_i} \text{ with}\nonumber \\
    &&\quad \quad SR(\ket{v_i}), SR(\ket{w_i}) \leq k \ \forall \, i \Big\}, \\
    R_k^{E,s}(\rho)& \defeq &\min_{\sigma : SN(\sigma) \leq k}\left\{s\geq 0\,:\, SN\left(\frac{\rho+s\sigma}{1+s}\right) \leq k \right\} \quad \ \ {}_{} \label{eq:standard_k_ent} \\
    R_k^{E,g}(\rho)& \defeq &\min_{\tau \in \cD_{mn}}\left\{s\geq 0\,:\, SN\left(\frac{\rho+s\tau}{1+s}\right) \leq k \right\}. \qquad {}_{} \label{eq:gen_k_ent}
\end{eqnarray}
It was shown in \cite[Theorem~10]{Reg17} that, for any pure state $\ket{v}\in \mathbb{C}^m \otimes \mathbb{C}^n$ and for any $k=1, \dots, \min\{m, n\}$, the equality $R_k^{E,g}(\ketbra{v}{v})=\|\ketbra{v}{v}\|_{\gamma, k}-1$ holds, and it was conjectured in \cite{Nconj} that the equality $R_k^{E,s}(\ketbra{v}{v})=\|\ketbra{v}{v}\|_{\gamma, k}-1$ holds.
 In addition to the trivial case of $k=\min\{m, n\}$ where the robustness $R_k^{E,s}$ of any state is equal to $0$, the conjecture was shown to be true when $k=1$ in \cite{VT99,Rud01}, since if $\ket{\lambda} := (\lambda_1,\lambda_2,\ldots,\lambda_r)^t$ is the vector of Schmidt coefficients of $\ket{v}$, then we have the explicit formulas
\begin{align*}
    \|\ketbra{v}{v}\|_{\gamma, 1} &= \left(\sum_{i=1}^r\lambda_i\right)^2 \quad \text{and} \\
    R_1^{E,s}(\ketbra{v}{v}) & = \left(\sum_{i=1}^r\lambda_i\right)^2 - 1.
\end{align*}
More generally, in \cite[Theorem~5.1]{Nconj} it was established that
\begin{align}\label{eq:ent_norm_rob_coh}
    \|\ketbra{v}{v}\|_{\gamma, k} &=  R_k^s(\ketbra{\lambda}{\lambda})+1\\
    &= \knorm{\ket{\lambda}}{k}^2,
\end{align}
where $R_k^s(\ketbra{\lambda}{\lambda})$ is given by the formula of Theorem~\ref{thm:main}. As an application of our main result (Theorem~\ref{thm:main}), we now show that this conjecture holds for all other values of $k$ as well.

\begin{theorem}\label{thm:robustness_entanglement}
    Let $\ket{v} \in \mathbb{C}^m \otimes \mathbb{C}^n$ be a pure state with non-zero Schmidt coefficients $\lambda_1, \lambda_2, \ldots, \lambda_r$ and define $\ket{\lambda} := (\lambda_1,\lambda_2,\ldots,\lambda_r)^t$. Then
\[
    R_k^{E,s}(\ketbra{v}{v}) = R_k^s(\ketbra{\lambda}{\lambda}) = \|\ketbra{v}{v}\|_{\gamma, k} - 1.
\]\end{theorem}
\begin{proof}
Assume without loss of generality that the Schmidt decomposition of $\ket{v}$ has the form $\ket{v} = \sum_{i=1}^n \lambda_i \ket{i}\otimes\ket{i}$ (if it does not have this form, we can multiply it by a local unitary to bring it into this form).

Since $R_k^{E,s}(\ketbra{v}{v}) \geq R_k^{E,g}(\ketbra{v}{v})$ trivially, we immediately have $R_k^{E,s}(\ketbra{v}{v}) \geq \|\ketbra{v}{v}\|_{\gamma, k} - 1$. Alternatively, this lower bound can be shown in the same way as the lower bound in the proof of Theorem~\ref{thm:main}: noting that the dual expression for the generalized robustness can be written as 
\begin{equation}\begin{aligned}
    R_k^{E,g}(\rho) = \max_{W \succeq 0}\big\{\tr(\rho W)\,:\,I-W\in \cV_k^{\circ}\big\}-1\\
\end{aligned}\end{equation}
where $\cV_k^\circ \defeq \{W = W^\dagger\,:\,\tr(W\rho)\geq 0 \ \ \forall \, \rho : SN(\rho) \leq k \}$, we can then choose $W = \mathbf{bb}^\dagger$, where $\mathbf{b}$ is a vector which achieves the maximum in the dual formulation of the norm $\|\ketbra{v}{v}\|_{\gamma, k}$, i.e. $\|\ketbra{v}{v}\|_{\gamma, k} = |\mathbf{b}^\dagger\ket{v}|^2$ (cf. \cite[Theorem~5.1]{Nconj}). Notice that such $\mathbf{b}$ can be chosen as $\mathbf{b} = \sum_i a_i \ket{i}\otimes\ket{i}$ where $\mathbf{a}\mathbf{a}^t$ is the feasible dual solution for $R_k^g(\ketbra{\lambda}{\lambda})$ established in Eq.~\eqref{eq:feasible_a}. This in particular allows for the construction of an explicit dual feasible solution which achieves the optimal value of $R_k^{E,s}$.

To show the upper bound, let $\delta^* \in \cI_k$ be a $k$-incoherent state that attains the minimum in $R_k^s(\ketbra{\lambda}{\lambda})$. That is,
    $$
        R_k^s(\ketbra{\lambda}{\lambda}) = \min\left\{s\geq 0\,:\,\frac{\ketbra{\lambda}{\lambda}+s\delta^*}{1+s}\in \cI_k\right\}.
    $$
    Since $\delta^* \in \cI_k$, we can write it as a convex combination of pure states
    $$
        \delta^* = \sum_j p_j \ketbra{v_j}{v_j},
    $$
    where each $\ket{v_j}$ has at most $k$ non-zero entries:
    $$
        \ket{v_j} = \sum_{i=1}^k c_{i,j}\ket{i_j}.
    $$
    If we define $\ket{w_j} := \sum_{i=1}^k c_{i,j}(\ket{i_j} \otimes \ket{i_j})$ then it is the case that $SR(\ket{v_j}) \leq k$ and thus the mixed state
    $$
        \sigma^* := \sum_j p_j \ketbra{w_j}{w_j}
    $$
    has $SN(\sigma^*) \leq k$. A calculation then reveals that
    \begin{align*}
        R_k^{E,s}(\ketbra{v}{v}) & = \min_{\sigma : SN(\sigma) \leq k}\left\{s\geq 0 : SN\left(\frac{\ketbra{v}{v}+s\sigma}{1+s}\right) \leq k \right\} \\
        & \leq \min\left\{s\geq 0\,:\, SN\left(\frac{\ketbra{v}{v}+s\sigma^*}{1+s}\right) \leq k \right\} \\
        & \leq \min\left\{s\geq 0\,:\, \frac{\ketbra{\lambda}{\lambda}+s\delta^*}{1+s} \in \cI_k \right\} \\
        & = R_k^s(\ketbra{\lambda}{\lambda}),
    \end{align*}
    where the final inequality comes from the fact that $\frac{\ketbra{\lambda}{\lambda}+s\delta^*}{1+s} \in \cI_k$ implies $SN\left(\frac{\ketbra{v}{v}+s\sigma^*}{1+s}\right) \leq k$.
\end{proof}

\section{VI. Conclusions and Discussion}

In this paper, we derived a formula for the standard robustnesses of $k$-coherence and $k$-entanglement on pure states that agrees with known formulas for the corresponding generalized robustnesses, thus resolving conjectures about both of these families of measures and providing computable expressions for them. As our proof was non-constructive in nature, we also presented a computational method based on linear programming that allows us to quickly compute the closest $k$-incoherent state or closest Schmidt number $k$ state.

\section{Acknowledgements}
N.J.\ was supported by NSERC Discovery Grant number RGPIN-2016-04003. C.-K.L.\ is an affiliate member of the Institute for Quantum Computing, University of Waterloo. He is an honorary professor of Shanghai University. His research was supported by USA NSF grant DMS 1331021, Simons Foundation Grant 351047, and NNSF of China Grant 11571220. S.P.\ was supported by NSERC Discovery Grant number 1174582, the Canada Foundation for Innovation (CFI) grant number 35711, and the Canada Research Chairs (CRC) Program grant number 231250. B.R.\ was supported by the European Research Council (ERC) under the Starting Grant GQCOP (Grant No.~637352). Part of the research was done when C.-K.L.\ and Y.-T.P.\ were visiting the Institute for Quantum Computing in the fall of 2017. They gratefully acknowledge the generous support of the institute. 

\appendix

\end{document}